\documentclass[12pt]{article}

\usepackage{amsmath, amssymb, amsthm}
\usepackage{graphicx}
\usepackage{booktabs}
\usepackage{multirow}
\usepackage{xcolor}
\usepackage{algorithm}
\usepackage{algorithmic}
\usepackage[authoryear,longnamesfirst]{natbib}
\usepackage{url}

\newtheorem{assumption}{Assumption}

\newtheorem{proposition}{Proposition}

\newcommand{\shorttitle}{Functional clustering via random projections}
\newcommand{\shortauthors}{S. Chakrabarty et al.}
\title{Two-stage Ensemble Clustering of Functional Data Using Random Projections}

\author{
  Sourav Chakrabarty\thanks{Corresponding author. Email: chakrabartysourav024@gmail.com}\\
  \normalsize Applied Statistics Unit, Indian Statistical Institute, Kolkata, India
  \and
  Anirvan Chakraborty\\
  \normalsize Department of Mathematics and Statistics,\\ Indian Institute of Science Education and Research Kolkata, India
  \and
  Shyamal K. De\\
  \normalsize Applied Statistics Unit, Indian Statistical Institute, Kolkata, India
}

\date{}

\begin{document}
\pagestyle{empty}

\maketitle

\begin{abstract}
We propose a computationally simple framework for clustering functional data based on Gaussian-process-generated random projections. In this approach, each curve is first projected onto a large collection of independent Gaussian process realizations. The resulting high-dimensional representations are clustered using the Mean Absolute Difference of Distances (MADD), a dissimilarity measure well suited for high-dimensional settings. A population-level analysis of this dissimilarity provides insight into how random projections help capture distributional differences between functional populations. We introduce a second stage of clustering to additionally leverage on data-driven projection directions. Thus, in Stage I, an initial clustering is obtained using a set of prespecified projection families. In Stage II, this partition is refined by constructing Gaussian random projections based on an estimated covariance operator that uses the first stage of cluster labels. Finally, a normalized cost function is used to select the optimal clustering among candidate solutions. The proposed clustering algorithm is broadly applicable to diverse functional data regimes including irregular and partially observed data. Through extensive simulations and real-data applications, we show that the proposed method achieves a high degree of accuracy and outperforms many of the state-of-the-art methods across a wide range of functional data settings.
\end{abstract}
\section{Introduction}
Functional data arise across a broad range of scientific disciplines, e.g., chemometrics (spectrometric curves), meteorology (temperature and precipitation profiles) and biomedicine (growth trajectories and fMRI signals). In practice, such data are observed on dense, sparse or irregular grids and are often partially observed over the domain. Clustering functional data presents intrinsic statistical challenges due to the infinite-dimensional nature of the data, which are elements of a separable Hilbert space such as $L^2[a,b]$. Consequently, the methodology required to cluster such data needs to be fundamentally different from the classical multivariate clustering techniques.

A common approach to clustering functional data is to define a distance between curves, such as the $L_2$ distance and its variants, and then apply algorithms like k-means or k-medoids (see \cite{chen2014optimally}, \cite{martino2019k}, \cite{chen2021clustering}). Another widely used strategy is to project the functional observations onto a finite-dimensional subspace and perform multivariate clustering. Such projections can be obtained either by expanding fully observed curves in an orthonormal basis (see \cite{chiou2007functional}, \cite{LUZLOPEZGARCIA2015231}, \cite{delaigle2019clustering}), or by representing discretely observed data through truncated basis expansions such as spline, principal component, or wavelet bases (see \cite{abraham2003unsupervised}, \cite{kayano2010functional} and \cite{10.1111/j.1541-0420.2012.01828.x}). In addition, several model-based methods have been proposed for functional clustering (see \cite{margaritella2021parameter}, \cite{tfun-paper},\cite{sasfun-paper}, \cite{akeweje2024learning}, \cite{leroy2022magma}), while some works use pseudo-density–based approaches (see \cite{JACQUES2013164} and \cite{FADP-paper}). More recently, deep learning methods, including autoencoders, have also been explored (see \cite{FNN-paper}, \cite{FAE-paper}, \cite{FAEclust-paper}).

A key advantage of using random projections is that the projected values can be computed regardless of the nature of the observation grid of each functional datum, whether the curves are observed on dense, sparse or irregular grids, or are partially observed. This flexibility arises because the observation grid of the generated random direction can be chosen to match that of the functional datum when computing the inner products associated with the projections. Consequently, our procedure naturally accommodates a wide range of observation grids without requiring prior smoothing, thereby avoiding potential smoothing bias that may adversely affect clustering performance. Another appealing feature of our approach is its flexibility with respect to the choice of Gaussian distributions used to generate the random projections. We apply the clustering algorithm under different choices of such distributions, evaluate their performance using a suitable clustering index, and then combine the results through an ensemble strategy to obtain the final clusters.

The main contributions and distinguishing features of the proposed clustering methodology are as follows. First, in contrast to existing projection-based approaches for functional clustering, which typically rely on a limited number of projection directions and consequently yield low-dimensional summaries of inherently infinite-dimensional objects, our method employs a large collection of random projections drawn from suitable Gaussian measures. This strategy enables the aggregation of information across a rich set of directions in the underlying function space, thereby inducing a high-dimensional representation that captures subtle distributional differences among functional observations. To the best of our knowledge, this constitutes the first systematic use of such high-dimensional representations derived from random projections in the context of functional clustering. Second, we include a data-adaptive projection scheme in the proposed methodology. Specifically, leveraging the initial clustering obtained from the first stage, we estimate cluster-specific principal component directions and construct a clustering index analogous to that used in the first stage. The final clustering is then selected based on the stage that yields the lower value of this index. This two-stage procedure effectively combines the strengths of an ensemble of random projections and those of the adaptive data-driven projections, which improves the final clustering performance. Third, owing to the flexibility of the random projection framework, the proposed method is broadly applicable across diverse functional data regimes, including densely observed, sparsely observed, irregularly sampled, and partially observed trajectories. Importantly, the procedure circumvents the need for any preliminary smoothing, thereby avoiding potential biases associated with smoothing-based pre-processing.

\section{High-dimensional random projections and MADD}
\label{MADD-section}
Let $\{X_1, \dots, X_n\}$ denote a set of $n$ functional observations, where each $X_i \in L^2(\mathcal{T})$, which is the space of real square-integrable functions defined on a compact interval $\mathcal{T} \subset \mathbb{R}$. The space $L^2(\mathcal{T})$ is equipped with the standard inner product $\langle f, g \rangle = \int_{\mathcal{T}} f(t) g(t) \, dt$ and the associated norm $\|\cdot\|$, which makes it a separable Hilbert space.

Consider a zero mean Gaussian process with covariance operator $C$, denoted by $\mathcal{GP}(0,C)$. Let $\{Z_{1}, \dots, Z_{M}\}$ denote $M$ independent random functions drawn from this process. For each functional observation $X_i$, $1 \leq i \leq n$,  we define the projected vector
\begin{align}
\label{proj_def}
\mathbf{X}_i^* = (\langle X_i, Z_{1} \rangle, \langle X_i, Z_{2} \rangle, \dots, \langle X_i, Z_{M} \rangle).
\end{align}
Conditional on $X_i$, the projected variables $\{\langle X_i, Z_{q} \rangle\}_{q=1}^M$ are independent and identically distributed (i.i.d.). Consequently, $\mathbf{X}_i^*$
has exchangeable components, yielding a finite-dimensional representation in which cluster-related differences accumulate as the number of projections $M$ increases.

The use of Gaussian-process-based random projections is motivated by the theory developed in \cite{cuesta2007sharp}. In particular, Theorem 4.1 of that paper establishes that for random functions $X$ and $Y$ on $L^2(\mathcal{T})$, equality of their distributions can be characterized through their projections along random directions (that is, random functions) drawn from a non-degenerate Gaussian process. More precisely, let $A = \{z \in L^2(\mathcal{T}) : \langle X,z\rangle \stackrel{d}{=} \langle Y,z\rangle\}$. Then, under mild moment conditions, $X \stackrel{d}{\neq} Y$ if and only if $\mathbb{P}(Z \in A^c) = 1$, where $Z$ is a Gaussian random function (or direction). Consequently, for any Gaussian random direction, the distributions of the projected $X$-samples and projected $Y$-samples will be different whenever the distributions of $X$ and $Y$ are different. From a practical perspective, we consider a large number, say $M$, of random projections in order to capture the distributional differences between $X$ and $Y$ in a more effective manner.

Once we project the functional data into a $M$-dimensional space, where $M$ is potentially large, clustering these new projected observations using a $k$-means type algorithm would require an appropriate notion of distance between any pair of the projected observations. However, the classical Euclidean distance would not be useful in this setup due to the well-known ``curse of dimensionality''. For clustering high-dimensional data, a notion of distance (MADD) was considered in \cite{sarkar2019perfect} which was shown to be effective for such data settings. Given a set of $n$ observations $\mathcal{D}=\{\mathbf{u}_1,\dots,\mathbf{u}_n\}$ with $\mathbf{u}_i \in \mathbb{R}^M$ for all $i=1,2,\dots,n$,
the Mean Absolute Difference of Distances (MADD) 
between two vectors $\mathbf{v}, \mathbf{w} \in \mathcal{D}$ 
is defined as
\begin{align}
\label{MADD_def}
&\rho(\mathbf{v}, \mathbf{w})
=
\frac
{\sum_{\mathbf{u} \in \mathcal{D} \setminus \{\mathbf{v}, \mathbf{w}\}}
\left|
d(\mathbf{v}, \mathbf{u})
-
d(\mathbf{w}, \mathbf{u})
\right|}{n-2},
\nonumber \\ &\mbox{where} \ \
d(\mathbf{v}, \mathbf{w})
=
\frac{1}{M}
\sum_{q=1}^{M}
\psi\!\left( \lvert v^{(q)} - w^{(q)} \rvert \right),
\end{align}
with $\mathbf{v}=(v^{(1)},\dots,v^{(M)})^\top$, $\mathbf{w}=(w^{(1)},\dots,w^{(M)})^\top$ and $\psi(t)=1-\exp(-t)$. Note that $\rho(\mathbf{v}, \mathbf{w})$ is defined relative to the fixed dataset $\mathcal{D}$, since it compares the distance profiles of $\mathbf{v}$ and 
$\mathbf{w}$ with respect to the remaining observations. It was proved in \cite{sarkar2019perfect} that the $k$-means clustering based on MADD achieves asymptotically perfect clustering performance when $M \rightarrow \infty$ and $n$ is fixed, which is the well-known high-dimensional low-sample size (HDLSS) setup. This motivates us to consider MADD as the choice of distance measure for the high-dimensional projected data. Another theoretical justification of selecting MADD in our clustering algorithms will be provided in the following subsection.

\subsection{A population-level characterization of MADD}
\label{MADD_theory}
Suppose that we have a collection of $n$ independent random vectors $\mathcal{U}=\{\mathbf{U}_1,\mathbf{U}_2,\dots,\mathbf{U}_n\}$ where each $\mathbf{U}_i$
takes values in $\mathbb{R}^M$ and follows one of the $K$ distributions $G_1,\ldots,G_K$. Define $d^*_{ab}=\mathbb{E}[d(\mathbf{V}, \mathbf{W})]$ where 
$\mathbf{V}$ and $\mathbf{W}$ are independent random vectors with ${\bf V} \sim G_a$ and ${\bf W} \sim G_b$ $(1 \leq a,b \leq K)$. Then the population version of MADD dissimilarity between the distributions $G_a$ and $G_b$ considered in \cite{sarkar2019perfect} is 
\begin{align}
\label{rho*-def}
\rho^*(G_a,\, G_b) = \frac{1}{n-2}\Big\{&(n_a-1)|d^*_{ab}-d^*_{aa}|+(n_b-1)|d^*_{ab}-d^*_{bb}| \nonumber \\&+\sum_{c=1;\ \substack{c \neq a,b}}^{K}n_{c}|d^*_{ac}-d^*_{bc}|\Big\},
\end{align}
\noindent where $n_r$ denotes the number of observations coming from $r$-th population $(r=1,2,\dots,K)$. It can be shown that
\[
    \rho^*(G_a,\, G_b)=\frac
{\sum_{\mathbf{U} \in \mathcal{U} \setminus \{\mathbf{V}, \mathbf{W}\}}
\left|
\mathbb{E}[d(\mathbf{V}, \mathbf{U})]
-
\mathbb{E}[d(\mathbf{W}, \mathbf{U})]
\right|}{n-2},
\]
\noindent
where $\mathbf{V} \stackrel{}{\sim} G_a$ and $\mathbf{W} \stackrel{}{\sim} G_b$.
Note that $\rho^*$ is obtained by replacing $d(\cdot,\cdot)$ appearing in the definition of $\rho$ by its expectation. In our framework, $\mathbf{V}=\mathbf{X}_r^*$ and $\mathbf{W}=\mathbf{X}_s^*$ for some $r,s \in \{1,\dots,n\}$ (see \eqref{proj_def}). Since the coordinates of ${\bf X}_r^*$ (also those of ${\bf X}_s^*$) are identically distributed, it follows that
\[
d^*_{ab}
=
\mathbb{E}\!\left[
\psi\!\left(
\big|
\langle X_r,Z_{1}\rangle
-
\langle X_s,Z_{1}\rangle
\big|
\right)
\right],
\]
where $\mathbf{X}_r^* \sim G_a$ and $\mathbf{X}_s^* \sim G_b$ independently.
Thus, $d({\bf X}_r^*,{\bf X}_s^*)$ as defined in \eqref{MADD_def} serves as an unbiased estimator of the corresponding expectation $d^*_{ab}$.

We next provide a theoretical justification for using the MADD dissimilarity on the randomly projected data in developing a $k$-means type clustering algorithm. We start with the following Carleman-type moment assumption on the absolute moments of a probability distribution $F$ on the separable Hilbert space $\mathcal H = L^2(\mathcal{T})$ (see \cite{cuesta2007sharp}).

\begin{assumption}
\label{assume-1}
The $r$-th order absolute moment $m_r = \int_{\mathcal H} \|x\|^r \, dF(x)$, is finite for every $r \ge 1$ and $\{m_r\}_{r \ge 1}$ satisfies the divergence condition $\sum_{r \ge 1} m_r^{-1/r} = \infty$.
\end{assumption}

\begin{proposition}
\label{MADD connection with MMD}
Consider a collection $X_1,\ldots,X_n$ of $\mathcal{H}$-valued random functions, where each random function follows one of the $K$ probability distributions $F_1,\ldots,F_K$ with $n_r (\geq 2)$ denoting the number of random functions following distribution $F_r$, $1 \leq r \leq K$. Suppose that Assumption \ref{assume-1} holds for each of $F_1,\ldots,F_K$. Let $Z_{1},\ldots,Z_{M}$ be i.i.d. random functions drawn from $\mathcal{GP}(0,C)$, independently of $X_1,\ldots,X_n$. Then
\[
\rho^*\!\left(
F^{proj}_{a},
F^{proj}_{b}
\right)=0
\quad\text{if and only if}\quad
F_a = F_b,
\]
where $F^{proj}_{r}$ is the distribution of $\mathbf{X}^* = (\langle X,Z_1\rangle,\dots,\langle X,Z_M\rangle)$, with $X\sim F_r$, $1 \leq r \leq K$.
\end{proposition}

\begin{proof}
Let $X_{a1},\ldots,X_{an_a}$ denote the sample functions following the $a$-th distribution $F_a$ for each $a = 1,\ldots,K$. Since each random function in the sample $\{X_1,\ldots,X_n\}$ follows one of the $K$ probability distributions $F_1,\ldots,F_K$, we have $\sum_{a=1}^K n_a=n$. It follows from equation \eqref{rho*-def} that
\begin{align} 
\label{population rho}
(n-2)\rho^*\!\left(
F^{proj}_{a}, F^{proj}_{b}
\right) 
&= (n_a-1)\big|\mathbb{E}[d(\mathbf{X}_{a1}^*, 
\mathbf{X}_{b1}^*)]
-\mathbb{E}[d(\mathbf{X}_{a1}^*, 
\mathbf{X}_{a2}^*)]\big| \nonumber\\ & \quad + (n_b-1)\big|\mathbb{E}[d(\mathbf{X}_{a1}^*, 
\mathbf{X}_{b1}^*)]
-\mathbb{E}[d(\mathbf{X}_{b1}^*, 
\mathbf{X}_{b2}^*)]\big| \nonumber\\
& \quad+ \sum_{c (\neq a,b) = 1}^{K} n_c\big|\mathbb{E}[d(\mathbf{X}_{a1}^*, \mathbf{X}_{c1}^*)]
-\mathbb{E}[d(\mathbf{X}_{b1}^*, 
\mathbf{X}_{c1}^*)]\big|,
\end{align}
where $\mathbf{X}_{a1}^* = (\langle X_{a1}, Z_{1} \rangle, \langle X_{a1}, Z_{2} \rangle, \dots, \langle X_{a1}, Z_{M} \rangle)$ as in \eqref{proj_def}.
It is easy to see that $\rho^*\!\left(F^{proj}_{a}, F^{proj}_{b}\right)=0$ if $F_a=F_b$. For proving the only if part, assume that $F_a\neq F_b$. In this case,
\begin{align*}
(n-2)\rho^*\!\left(
F^{proj}_{a}, F^{proj}_{b}
\right) &= (n_a-1)|A|+(n_b-1)|B|+R, \quad \mbox{where} \\
A &= \mathbb{E}[d(\mathbf{X}_{a1}^*, \mathbf{X}_{b1}^*)]
-\mathbb{E}[d(\mathbf{X}_{a1}^*, \mathbf{X}_{a2}^*)], \\
B &= \mathbb{E}[d(\mathbf{X}_{a1}^*, \mathbf{X}_{b1}^*)]
-\mathbb{E}[d(\mathbf{X}_{b1}^*, \mathbf{X}_{b2}^*)], \quad \mbox{and}\\
R &= \sum_{c (\neq a,b) = 1}^{K} n_c|\mathbb{E}[d(\mathbf{X}_{a1}^*, 
\mathbf{X}_{c1}^*)]-\mathbb{E}[d(\mathbf{X}_{b1}^*, 
\mathbf{X}_{c1}^*)]|.
\end{align*}
\noindent
Note that $R \geq 0$. Since $Z_1, \dots, Z_M$ are 
i.i.d., $\langle X, Z_q\rangle$ is identically distributed 
across $q = 1, \dots, M$ for any $X \in \mathcal{H}$. Using this fact and the definition of $d(\cdot,\cdot)$, we obtain
\begin{align*}
A+B
&= \mathbb{E}\big[\exp(-|\langle X_{a1},Z_{1}\rangle
-\langle X_{a2},Z_{1}\rangle|)\big]
+\mathbb{E}\big[\exp(-|\langle X_{b1},Z_{1}\rangle
-\langle X_{b2},Z_{1}\rangle|)\big]
\\
& \quad - 2\,\mathbb{E}\big[\exp(-|\langle X_{a1},Z_{1}\rangle
-\langle X_{b1},Z_{1}\rangle|)\big].
\end{align*}
\noindent
For any real-valued random variable $W$, we denote its law (i.e., distribution) by $\mathcal{L}(W)$. Define
\[
g(z)
=
\mathrm{MMD}^2_k\!\left(
\mathcal{L}(\langle X_{a1},z\rangle),
\mathcal{L}(\langle X_{b1},z\rangle)
\right),
\]
where $\mathrm{MMD}^2_k$ denotes the squared maximum mean 
discrepancy measure associated with the Laplace kernel
$k(u,v)=\exp(-|u-v|)$ for $u, v \in \mathbb{R}$, and $z \in \mathcal{H}$ is treated as 
a fixed deterministic projection direction. Let $\widetilde{Z}$ denote a 
generic random direction having the same distribution as $Z_1$, 
independent of all functional observations $X_1, \dots, X_n$. Then, it follows from the properties of conditional expectations that $A+B
=
\mathbb{E}_{\widetilde{Z}}\!\left[
g(\widetilde{Z})
\right].$
Since $F_a \neq F_b$ and both distributions satisfy Assumption \ref{assume-1}, the assumptions of Theorem 4.1 of \cite{cuesta2007sharp} are satisfied and consequently,
\[
\mathbb{P}\left[{\widetilde{Z}} \in\!\left\{ z \in \mathcal{H} : 
\mathcal{L}(\langle X_{a1},z\rangle)
\neq
\mathcal{L}(\langle X_{b1},z\rangle)
\right\}\right]=1.
\]
This implies that for almost every Gaussian direction $z$, the projected 
one-dimensional distributions of $F_a$ and $F_b$ differ.

Since the Laplace kernel $k(u,v)=\exp(-|u-v|)$ is a characteristic kernel
on $\mathbb{R}$ (see \cite{JMLR:v12:sriperumbudur11a}), it follows that 
\[
\mathbb{P}\left[{\widetilde{Z}} \in \!\left\{ z \in \mathcal{H} : 
g(z) > 0
\right\}\right] = 1, \text{ and consequently, } A + B = \mathbb{E}_{\widetilde{Z}}[g(\widetilde{Z})] > 0.
\]
\noindent
Since $A + B > 0$, we have 
$(n_a-1)|A| + (n_b-1)|B| > 0$, where the last 
inequality uses $n_r \ge 2$ for all $r = 1, \dots, K$, which 
ensures both coefficients $(n_a - 1)$ and $(n_b - 1)$ are strictly positive. We complete the proof by noting that $R \geq 0$ implies 
\[
(n_a-1)|A|+(n_b-1)|B|+R > 0, \text{ and hence, } \rho^*\!\left(
F^{proj}_{a},
F^{proj}_{b}
\right)>0.
\]
\end{proof}

This result shows that the population version of the MADD measure, computed between the distributions of projected functional data (via Gaussian random projections), preserves distributional differences whenever the underlying $\mathcal{H}$-valued functional data distributions are distinct. Consequently, it provides theoretical justification for using random projections in conjunction with MADD-based dissimilarity for clustering functional data.

\section{Clustering methodology: regular functional data}
\label{sec:methodology}
Recall that we observe a sample of $n$ functional observations $\{X_1,\ldots,X_n\}$ where each $X_i \in L^2(\mathcal{T})$ and each curve is observed over a finite discrete grid, say $\mathcal{I} \subset \mathcal{T}$, which is common for all observations. Each sample function is transformed into an $M$-dimensional vector using random projections based on a Gaussian process $\mathcal{GP}(0,C)$ as in the previous section.

\subsection{First stage of clustering via random projections}
The goal is to perform a $k$-means clustering using a MADD-based modified objective function applied on the projected data $\mathcal{D} = \{{\bf X}_1^*,\ldots,{\bf X}_n^*\}$ as in \eqref{proj_def}. To this end, define the following measure of within-cluster variance for a partition $\mathcal{S}=\{S_1,\dots,S_K\}$  of $\{1,2,\ldots,n\}$ into $K$ non-empty subsets: 
\begin{align}
\label{eq:within_cost}
\Phi^{(1)}(\mathcal{S})
=
\sum_{r=1}^{K}
\frac{1}{2|S_r|}
\sum_{i, j \in S_r}
\rho^2(\mathbf{X}_i^*, \mathbf{X}_j^*),
\end{align}
where $\rho(\cdot,\cdot)$ is the MADD dissimilarity defined in \eqref{MADD_def}. Although \cite{sarkar2019perfect} minimized the above objective function to obtain the optimal partition, we need to consider a scale-invariant version of the above function. This modification will be crucial when we compare the objective functions for various choices of the number of projections $M$ while developing the ensemble clustering in a later section. The modified objective function is defined as
\begin{eqnarray}
\label{eq:within_cost_new}
\widetilde{\Phi}^{(1)}(\mathcal{S}) &=& \Phi^{(1)}(\mathcal{S})/T^{(1)}, \ \ \mbox{where} \ \
T^{(1)} = \frac{1}{2n}\sum_{i,j=1}^n \rho^2(\mathbf{X}_i^*, \mathbf{X}_j^*)
\end{eqnarray}
represents the total variation among the projected data. The clusters are obtained in the first stage by performing a $k$-means type algorithm which chooses a partition $S$ that minimizes the objective function given in \eqref{eq:within_cost_new} for a prefixed value of $K$.

\subsection{Second stage of clustering via data-driven projections}
\label{subsec:two_stage}
Denote the clusters obtained in the first stage of clustering by $I_1, \ldots, I_K$, which is a partition of $\{1,2,\ldots,n\}$. We compute the cluster-wise sample mean functions as
\[
\widehat{\mu}_k(t) = \frac{1}{|I_k|} \sum_{i \in I_k} X_i(t), \quad k=1,\dots,K, \ t \in \mathcal{I}.
\]
We then center the curves according to their assigned clusters. For each $1 \leq i \leq n$,
\[
\widetilde{X}_i(t) = X_i(t) - \widehat{\mu}_{k}(t), \quad \mbox{if} \ i \in I_k, \ k=1,\dots,K, \ t \in \mathcal{I}.
\]
Using the centered sample $\{\widetilde{X}_i : i \in I_k, \ 1 \leq k \leq K\}$, we estimate the pooled covariance operator $\widehat{C}$ from the $K$ estimated cluster-wise covariance operators and obtain the estimated eigenpairs 
$\{(\widehat{\lambda}_j,\widehat{\phi}_j)\}_{j\ge1}$ (see, for example, \cite{horvath2012inference}). Conditional on $\widehat{C}$, we generate $M$ independent Gaussian random variables from $\mathcal{GP}(0,\widehat{C})$ as follows: 
\[
W_{q}(t) = \sum_{j\ge1} \sqrt{\widehat{\lambda}_j}\, \xi_{j q}\, \widehat{\phi}_j(t), \ 1 \leq q \leq M, \ t \in \mathcal{I},
\]
where $\xi_{j q} \stackrel{\text{i.i.d.}}{\sim} \mathcal{N}(0,1)$ for all $j \geq 1$ and $ 1\leq q \leq M$. Next, each original curve $X_i$ is projected onto $\{W_{q}\}_{q=1}^M$ to obtain
\[
\mathbf{X}_{i}^{**}
=
\big(\langle X_i, W_{1}\rangle,\dots,\langle X_i, W_{M}\rangle\big).
\]
To obtain the second stage clusters, we now apply a similar $k$-means-type clustering procedure as in the first stage to the data $\{\mathbf{X}_{i}^{**}\}_{i=1}^n$. Thus, the second stage clusters are obtained by minimizing the cost function
\begin{eqnarray}
\label{eq:within_cost_new_stage2}
&\widetilde{\Phi}^{(2)}(\mathcal{S}) = \Phi^{(2)}(\mathcal{S})/T^{(2)}, 
\mbox{with} \ \ \Phi^{(2)}(\mathcal{S}) = \sum_{r=1}^{K} \frac{1}{2|S_r|} \sum_{i, j \in S_r}
\rho^2(\mathbf{X}_i^{**}, \mathbf{X}_j^{**}) \ \ \mbox{and} \nonumber\\  
&T^{(2)} = \frac{1}{2n}\sum_{i,j=1}^n \rho^2(\mathbf{X}_i^{**}, \mathbf{X}_j^{**}), 
\end{eqnarray}
over partitions $\mathcal{S}=\{S_1,\dots,S_K\}$ of $\{1,2,\ldots,n\}$ into $K$ non-empty subsets.

\subsection{Ensemble clustering: aggregation over projection directions and number of projections}
The first stage of clustering described earlier depends both on the choice of the covariance operator $C$ in the $\mathcal{GP}(0,C)$ process considered earlier and the number of random projections $M$. These choices also crucially determine the performance of the second stage of clustering. In order to deal with this issue, we propose an ensemble procedure where we apply the two-stage clustering algorithm for various choices of $C$ and $M$. Further details on the specific choices of $C$ and $M$ are provided in Section \ref{sec:numerical}. Note that different choices of $C$ and $M$ produce finite-dimensional projections which differ in scale, and hence the corresponding within-cluster costs $\Phi^{(1)}(\mathcal{S})$ (respectively, $\Phi^{(2)}(\mathcal{S})$) appearing in \eqref{eq:within_cost} (respectively, \eqref{eq:within_cost_new_stage2}) are not comparable across these choices. Consequently, to ensure scale-invariance, we normalize $\Phi^{(i)}(\mathcal{S})$ by the total dispersion $T^{(i)}$ associated with the $i$-th stage of clustering, $i \in \{1,2\}$ (see \eqref{eq:within_cost_new} and \eqref{eq:within_cost_new_stage2}). For each combination of $C$ and $M$, the two-stage clustering algorithm is implemented and the minimum value, say $v_{C,M}$, of the scale-invariant cost functions $\widetilde{\Phi}^{(1)}(\mathcal{S})$ and $\widetilde{\Phi}^{(2)}(\mathcal{S})$ is obtained. The final clusters are those that correspond to the minimizer of $v_{C,M}$ over the various choices of $C$ and $M$. The complete clustering procedure is summarized in Algorithm~\ref{alg:grpm}. 

\begin{algorithm}
\renewcommand{\thealgorithm}{}
\caption{\textbf{TERP:} Two-stage Ensemble Random Projection clustering}
\label{alg:grpm}
\small
\begin{algorithmic}

\REQUIRE Functional data $\{X_i\}_{i=1}^n$, number of clusters $K$, projection families $\{\mathcal{GP}(0,C_l)\}_{l=1}^L$, set of number of projections $\mathcal{M}$

\vspace{0.1cm}

\FOR{$l=1$ to $L$}

\FOR{each $M \in \mathcal{M}$}

\vspace{0.05cm}

\STATE \textbf{Stage I}

\STATE Generate $\{Z_{l,q}\}_{q=1}^{M} \stackrel{\text{i.i.d.}}{\sim} \mathcal{GP}(0,C_l)$

\STATE For $i=1,2,\dots,n$, compute the projected vectors
\[
\mathbf{X}_{i,l}^{*}
=
\big(\langle X_i,Z_{l,1}\rangle,\dots,\langle X_i,Z_{l,M}\rangle\big)
\]

\STATE Obtain partition
$\mathcal P_{l,M}^{(1)}$
which minimizes \eqref{eq:within_cost_new}

\STATE Compute normalized cost $\widetilde{\Phi}_{l,M}^{(1)}
=
\widetilde{\Phi}^{(1)}(\mathcal P_{l,M}^{(1)})$

\vspace{0.05cm}

\STATE \textbf{Stage II}

\STATE Let the partition sets in $\mathcal P_{l,M}^{(1)}$ be denoted by $I_1,\ldots,I_K$

\STATE Compute cluster means for each $1 \leq k \leq K$ given by
\[
\widehat\mu_k = \frac{1}{|I_k|} \sum_{i \in I_k} X_i
\]

\STATE Center the curves cluster-wise: $\widetilde X_i = X_i - \widehat\mu_{k(i)}$, where $k(i)$ denotes the cluster label of $X_i$ under $\mathcal P_{l,M}^{(1)}$

\STATE Estimate pooled eigenpairs $(\widehat\lambda_j,\widehat\phi_j)$ using $\widetilde X_1,\ldots,\widetilde X_n$

\STATE Generate i.i.d. samples $W_{l,q} = \sum_{j \geq 1} \sqrt{\widehat\lambda_j} \xi_{jq} \widehat\phi_j$, $1 \leq q \leq M$, where $\xi_{jq} \stackrel{\text{i.i.d.}}{\sim} \mathcal{N}(0,1)$

\STATE For $i=1,2,\dots,n$, compute the projected vectors
\[
\mathbf{X}_{i,l}^{**}
=
\big(\langle X_i,W_{l,1}\rangle,\dots,\langle X_i,W_{l,M}\rangle\big)
\]

\STATE Obtain partition $\mathcal P_{l,M}^{(2)}$ which minimizes \eqref{eq:within_cost_new_stage2}

\STATE Compute normalized cost $\widetilde{\Phi}_{l,M}^{(2)} = \widetilde{\Phi}^{(2)}(\mathcal P_{l,M}^{(2)})$

\STATE Compute $v_{l,M} = \min\left(\widetilde{\Phi}_{l,M}^{(1)},\widetilde{\Phi}_{l,M}^{(2)}\right)$

\ENDFOR

\ENDFOR

\vspace{0.1cm}

\STATE \textbf{Final selection}
\begin{eqnarray}
\label{min-cost}
(l^\ast,M^\ast)
=
\underset{1 \leq l \leq L,\ M \in \mathcal{M}}{\arg\min}
\;
v_{l,M},
\quad 
s^\ast
=
\arg\min_{s \in \{1,2\}} \widetilde{\Phi}_{l^\ast,M^\ast}^{(s)}
\end{eqnarray}
\RETURN $\widehat{\mathcal S} = \mathcal P_{l^\ast,M^\ast}^{(s^\ast)}$
\ENSURE Final partition $\widehat{\mathcal S}$
\end{algorithmic}
\end{algorithm}

\section{Clustering methodology: irregular and fragmented functional data}
\label{sec:methodology-irregular}
The steps described in Algorithm \ref{alg:grpm} are applicable to functional data that are either fully observed or observed over a common dense grid. However, in many applications, each functional curve $X_i$ may be observed over a curve-specific grid $\mathcal{I}_i = \{t_{i1}, t_{i2}, \ldots, t_{im_i}\} \subset \mathcal{T}$ for $1 \leq i \leq n$. When this observation grid consists of a discrete non-contiguous set of points in the domain $\mathcal{T}$, the resulting functional data are referred to as irregularly observed. In contrast, when one or more contiguous segments (sub-intervals) of $\mathcal{T}$ are missing from the observation grids, the resulting functional data are termed partially observed or fragmented. Accommodating such forms of functional data requires partial modifications to the clustering algorithm. Note that the first stage of clustering remains unaltered because the inner products in the definition of ${\bf X}_i^*$ in \eqref{proj_def} can be computed using appropriate Riemann sums over the observation grid of $X_i$, with each projection direction $Z_q$ evaluated only on that same grid. However, the second stage of clustering needs to be modified. The most common approach is to apply the PACE procedure \citep{yao2005functional} to each of the clusters obtained in the first stage in order to estimate the pooled covariance operator required for the random projections in the second stage. The complete clustering procedure is summarized in Algorithm~\ref{alg:grpm_irreg}. 

\begin{algorithm}
\renewcommand{\thealgorithm}{}
\caption{\textbf{TERPM:} Modified TERP algorithm for irregular and fragmented functional data}
\label{alg:grpm_irreg}
\small
\begin{algorithmic}

\REQUIRE Functional data 
         $\{X_{i}(t) : t \in \mathcal{I}_i\}$ for $i=1,\dots,n$,
         number of clusters $K$,
         projection families $\{\mathcal{GP}(0,C_l)\}_{l=1}^L$,
         set of number of projections $\mathcal{M}$

\vspace{0.1cm}

\FOR{$l=1$ to $L$}

\FOR{each $M \in \mathcal{M}$}

\vspace{0.05cm}

\STATE \textbf{Stage I}

\STATE Generate $\{Z_{l,q}\}_{q=1}^{M} \stackrel{\text{i.i.d.}}{\sim} \mathcal{GP}(0,C_l)$

\STATE For $i=1,2,\dots,n$, compute the projected vectors
\[
\mathbf{X}_{i,l}^{*}
=
\big(\langle X_i,Z_{l,1}\rangle,\dots,\langle X_i,Z_{l,M}\rangle\big)
\]

\STATE Obtain partition $\mathcal P_{l,M}^{(1)}$ 
       which minimizes \eqref{eq:within_cost_new}

\STATE Compute normalized cost 
       $\widetilde{\Phi}_{l,M}^{(1)} 
       = \widetilde{\Phi}^{(1)}(\mathcal P_{l,M}^{(1)})$

\vspace{0.05cm}

\STATE \textbf{Stage II}

\STATE Let the partition sets in $\mathcal P_{l,M}^{(1)}$ 
       be denoted by $I_1,\ldots,I_K$

\STATE Estimate cluster mean functions for each $1 \leq k \leq K$ 
       using the PACE procedure applied to 
       $\{X_{i}(t) : t \in \mathcal{I}_i\}_{i \in I_k}$

\STATE Center the curves cluster-wise at the observed time points:
\[
\widetilde X_{i}(t) 
= 
X_{i}(t) - \widehat\mu_{k(i)}(t), 
\quad 
t \in \mathcal{I}_i, \ i=1,\dots,n,
\]
where $k(i)$ denotes the cluster label of $X_i$ 
under $\mathcal P_{l,M}^{(1)}$

\STATE Estimate pooled eigenpairs $(\widehat\lambda_j, \widehat\phi_j)$ 
       by applying PACE to the combined centered dataset 
       $\{\widetilde X_{i}(t) : t \in \mathcal{I}_i\}_{i=1,\dots,n}$

\STATE Generate i.i.d. samples 
       $W_{l,q} = \sum_{j \geq 1} \sqrt{\widehat\lambda_j}\, 
       \xi_{jq}\, \widehat\phi_j$, $\quad 1 \leq q \leq M$,
       where 
       $\xi_{jq} \stackrel{\text{i.i.d.}}{\sim} \mathcal{N}(0,1)$

\STATE For $i=1,2,\dots,n$, compute the projected vectors
\[
\mathbf{X}_{i,l}^{**}
=
\big(\langle X_i,W_{l,1}\rangle,\dots,\langle X_i,W_{l,M}\rangle\big),
\]

\STATE Obtain partition $\mathcal P_{l,M}^{(2)}$ 
       which minimizes \eqref{eq:within_cost_new_stage2}

\STATE Compute normalized cost 
       $\widetilde{\Phi}_{l,M}^{(2)} 
       = \widetilde{\Phi}^{(2)}(\mathcal P_{l,M}^{(2)})$

\STATE Compute $v_{l,M} = \min\left(\widetilde{\Phi}_{l,M}^{(1)},\widetilde{\Phi}_{l,M}^{(2)}\right)$

\ENDFOR

\ENDFOR

\vspace{0.1cm}

\STATE \textbf{Final selection}
\begin{eqnarray}
\label{min-cost-irreg}
(l^\ast,M^\ast)
=
\underset{1 \leq l \leq L,\ M \in \mathcal{M}}{\arg\min}
\;
v_{l,M},
\quad 
s^\ast
=
\arg\min_{s \in \{1,2\}} \widetilde{\Phi}_{l^\ast,M^\ast}^{(s)}
\end{eqnarray}

\RETURN $\widehat{\mathcal S} = \mathcal P_{l^\ast,M^\ast}^{(s^\ast)}$

\ENSURE Final partition $\widehat{\mathcal S}$

\end{algorithmic}
\end{algorithm}

\section{Numerical Studies}
\label{sec:numerical}
We now compare the performance of Algorithms \ref{alg:grpm} and \ref{alg:grpm_irreg} with various competitors across diverse functional data settings. We compare with the following methods: \texttt{PC} \cite{delaigle2019clustering}, \texttt{F1} and \texttt{F2} \cite{FADP-paper}, \texttt{$\text{D}_{\text{H}}$} and \texttt{$\text{D}_{\text{PC}}$} \cite{delaigle2019clustering}, \texttt{MC} \cite{leroy2022magma}, \texttt{sf} \cite{sasfun-paper} and \texttt{GPmix} \cite{akeweje2024learning}. The \texttt{PC} method was implemented using the \texttt{fdapace} package in CRAN. Implementations of the \texttt{$\text{D}_{\text{H}}$} and \texttt{$\text{D}_{\text{PC}}$} procedures are available at \url{http://researchers.ms.unimelb.edu.au/~aurored}. The \texttt{F1} and \texttt{F2} methods are implemented using the \texttt{FADPclust} package, the \texttt{sf} method is implemented via the \texttt{sasfunclust} package and the \texttt{MC} method is implemented using the \texttt{MagmaClustR} package in CRAN. The code for implementation of the \texttt{GPmix} is available at \url{https://github.com/EAkeweje/GPmix}. 

When evaluating the performance of Algorithm \ref{alg:grpm_irreg} alongside competing methods for irregularly observed and fragmented functional data, all clustering procedures are applied directly to the raw observed data without any preliminary smoothing or explicit trajectory reconstruction. To ensure a fair comparison, we consider only those competing methods that can be implemented directly on irregularly observed or fragmented functional data without requiring a prior smoothing step. Consequently, only the \texttt{PC}, \texttt{sf}, and \texttt{MC} methods are included in the comparison. For both Algorithms \ref{alg:grpm} and \ref{alg:grpm_irreg}, the inner products involved in the projection steps at each stage are approximated via numerical quadrature using the available observation points for the corresponding curves.

Random projection directions are generated as independent realizations from Gaussian processes defined on $[0,1]$. The following $L = 6$ Gaussian processes are considered:
{
\begin{align*}
&1. \ BM(t) = \mbox{standard Brownian motion}, 
\\
&2. \ BB(t) = \mbox{standard Brownian bridge},
\\
&3. \ {GP}_{\mathrm{Haar\text{-}poly}}(t)
= \sum_{j=0}^{J_{\max}}
   \sum_{k=1}^{2^j}
   \xi_{j,k}\,(2^j+k)^{-\alpha/2}  \psi_{j,k}(t),
\\
&4. \ {GP}_{\mathrm{Haar\text{-}exp}}(t)
= \sum_{j=0}^{J_{\max}}
   \sum_{k=1}^{2^j}
   \xi_{j,k}\,
   \exp\!\left(-\frac{2^j+k^2}{20000}\right) \psi_{j,k}(t),
\\
&5. \ {GP}_{\mathrm{Fourier\text{-}poly}}(t)
= \sum_{k=1}^{K_{\max}}
   k^{-\alpha/2}
   \{
   \eta_k\,\sin(2\pi k t)
   + \tilde{\eta}_k\,\cos(2\pi k t)
   \},
\\
&6. \ {GP}_{\mathrm{Fourier\text{-}exp}}(t)
= \sum_{k=1}^{K_{\max}}
   \exp\!\left(-\frac{k^2}{20000}\right)
   \{
   \eta_k\,\sin(2\pi k t)
   + \tilde{\eta}_k\,\cos(2\pi k t)
   \}.
\end{align*}
}
Here, $\xi_{j,k}\stackrel{\text{i.i.d.}}{\sim}\mathcal{N}(0,1)$, $\eta_k\stackrel{\text{i.i.d.}}{\sim}\mathcal{N}(0,1)$, $\tilde{\eta}_k\stackrel{\text{i.i.d.}}{\sim}\mathcal{N}(0,1)$ with $J_{\max}=4$, $K_{\max}=40$ and $\alpha = 2$. The $\{\psi_{j,k}\}$ denote the Haar wavelet basis on $[0,1]$. We consider exponential as well as polynomial decay of the eigenvalues to accommodate both smooth and non-smooth paths of the Gaussian projection directions. Haar-based projections are well suited for capturing localized features, while Fourier-based projections emphasize global smooth structure.

For each of the above six Gaussian processes, Algorithms \ref{alg:grpm} and \ref{alg:grpm_irreg} are implemented with the choice of the number of random projections $M \in \mathcal{M} = \{10, 50, 100, 500, 1000\}$. So, the ensemble clustering is based on $6 \times 5 = 30$ different combinations of Gaussian processes and $M$. It is observed in our simulations that the clustering performance does not change significantly if we increase $M$ beyond $1000$. Clustering performance is measured using the average Rand index \cite{rand1971} over $100$ independent Monte-Carlo repetitions. The Rand index takes values in $[0,1]$ with larger values indicating greater concordance between the estimated and the true partitions of the data. In all tables, the largest Rand index in each row is highlighted in \textbf{bold} and the second-largest in \textit{italics}. While implementing the competing methods, we observed that, in some instances (see Tables \ref{tab:comparison_real data}, \ref{tab:comparison_irregular_real_data}, and \ref{tab:comparison_fragmented_real_data}), certain methods either failed to converge or did not return any results. Such instances are indicated by hyphens in the corresponding tables.

\subsection{Analysis of simulated datasets: regular functional data} \label{simulation studies:reg}
For investigating the behaviour of Algorithm \ref{alg:grpm} and its competitors, we generate functional data on a common discrete grid of 100 equi-spaced points in the interval $\mathcal{T} = [0,1]$. Let $n_l$ denotes the number of observations drawn from the $l$-th population and $X_{l}(t)$ denotes an observation from the $l$-th population for $l =1,\ldots, K$. We consider six two-population models (Models 1--6) and four three-population models (Models 7--10).

\medskip

\noindent \textbf{Model 1} [Location change]:  $X_{l}(t) = \sum_{j=1}^{40}\sqrt{\theta_{j}}Z_{j}\phi_{j}(t) + \mu_{l}(t)$, $l = 1, 2$,
where for each $j$, $\phi_{j}(t) = \sqrt{2}\sin{(jt\pi)}$, $\theta_{j} = j^{-1.05}$, $Z_{j} \overset{\mathrm{i.i.d.}}{\sim} \mathcal{N}(0,1)$ for $j = 1, 2, \dots, 40$, $\mu_1(t) = 2(t^2-1/3)$ and $\mu_2(t) = 0$.

\smallskip

\noindent \textbf{Model 2} [Change in eigenvalues]:    $X_{l}(t) = \sum_{j=1}^{40}\sqrt{\theta_{jl}}Z_{j}\phi_{j}(t)$, $l = 1, 2$,
where for $j=1,2,\dots,40$, $\phi_{j}(t) = \sqrt{2}\sin{(jt\pi)}$, $\theta_{j1} = j^{-2}, \theta_{j2} = 2^{-j}$, $Z_{j} \overset{\mathrm{i.i.d.}}{\sim} \mathcal{N}(0,1)$. 

\smallskip

\noindent \textbf{Model 3} [Location change]:  $X_{l}(t) = \sum_{j=1}^{40}\sqrt{\theta_{j}}Z_{j}\phi_{j}(t) + \mu_{l}(t)$, $l = 1, 2$,
where for each $j$, $\phi_{j}(t) = \sqrt{2}\sin{(2jt\pi)}$, $\theta_{j} = j^{-1.05}$, $Z_{j} \overset{\mathrm{i.i.d.}}{\sim} t_{3}/\sqrt{3}$ for $j = 1, 2, \dots, 40$, $\mu_1(t) = \sqrt{|t-0.5|}$ and $\mu_2(t) = 0$.

\smallskip

\noindent \textbf{Model 4} [Change in eigenvalues]:    $X_{l}(t) = \sum_{j=1}^{40}\sqrt{\theta_{jl}}Z_{j}\phi_{j}(t)$, $l = 1, 2$,
where for $j=1,2,\dots,40$, $\phi_{j}(t) = \sqrt{2}\sin{(jt\pi)}$, $\theta_{j1} = j^{-2}, \theta_{j2} = e^{-j}$, $Z_{j} \overset{\mathrm{i.i.d.}}{\sim} t_{3}/\sqrt{3}$. 

\smallskip

\noindent \textbf{Model 5} [Location change]: $X_{l}(t) = \sum_{j=1}^{40}(\sqrt{\theta_{j}}Z_{j}+\mu_{lj})\phi_{j}(t)$, $l = 1, 2$,
where for each $j$, $\phi_{j}(t) = \sqrt{2}\sin{(jt\pi)}$, $\theta_{j} = j^{-2}$, $Z_{j} \overset{\mathrm{i.i.d.}}{\sim} \mathcal{N}(0,1)$ for $j = 1, 2, \dots, 40$, $\mu_{lj} = 0$ for $j>4$ for each $l=1,2$, $(\mu_{11},\mu_{12},\mu_{13},\mu_{14})=(0,-0.5,1,-0.5)$ and $(\mu_{21},\mu_{22},\mu_{23},\mu_{24})=(0,-0.75,0.75,-0.75)$.

\smallskip

\noindent \textbf{Model 6} [Location change]: $X_{1}(t)=BB(t), X_{2}(t)=BB(t)+\sqrt{t}$.

\smallskip

\noindent \textbf{Model 7} [Location change]:
$X_l(t) = \sum_{j=1}^{40} \sqrt{\theta_j} Z_j \phi_j(t) + \mu_l(t), 
\quad l = 1,2,3,
$ where $\phi_j(t) = \sqrt{2}\sin(j\pi t)$, $\theta_j = j^{-1.05}$ and 
$Z_j \overset{\mathrm{i.i.d.}}{\sim} \mathcal{N}(0,1)$ for $j = 1,2,\dots,40$. 
The mean functions are given by
$\mu_1(t) = 2\left(t^2 - \frac{1}{3}\right), \quad 
\mu_2(t) = 0, \quad 
\mu_3(t) = -2\left(t^2 - \frac{1}{3}\right)$.

\smallskip

\noindent \textbf{Model 8} [Change in eigenvalues]:
$X_l(t) = \sum_{j=1}^{40} \sqrt{\theta_{jl}} Z_j \phi_j(t)$, 
$l = 1,2,3,$
where $\phi_j(t) = \sqrt{2}\sin(j\pi t)$ and 
$Z_j \overset{\mathrm{i.i.d.}}{\sim} \mathcal{N}(0,1)$. The eigenvalues are given by
$\theta_{j1} = j^{-2}, \quad 
\theta_{j2} = 4^{-j}, \quad 
\theta_{j3} = e^{-j}, 
\quad j = 1,2,\dots,40$.

\smallskip

\noindent \textbf{Model 9} [Location change]:
$X_l(t) = \sum_{j=1}^{40} \left( \sqrt{\theta_j} Z_j + \mu_{lj} \right)\phi_j(t)$,
$l = 1,2,3$,
where $\phi_j(t) = \sqrt{2}\sin(j\pi t)$, $\theta_j = j^{-2}$ and 
$Z_j \overset{\mathrm{i.i.d.}}{\sim} \mathcal{N}(0,1)$. 
The coefficient shifts satisfy $\mu_{lj} = 0$ for $j > 4$ with $(\mu_{11},\mu_{12},\mu_{13},\mu_{14}) = (0,-0.5,1,-0.5)$, $(\mu_{21},\mu_{22},\mu_{23},\mu_{24}) = (0,-0.75,0.75,-0.75)$ and $(\mu_{31},\mu_{32},\mu_{33},\mu_{34}) = (0,-1.0,0.5,-1.0)$.

\smallskip

\noindent \textbf{Model 10} [Location change]:
$X_1(t) = BB(t), \quad 
X_2(t) = BB(t) + \sqrt{t}, \quad 
X_3(t) = BB(t) - \sqrt{t}$,
where $BB(t)$ denotes a standard Brownian bridge on $[0,1]$.

In all of the above models except for Models 3 and 4, $Z_{j}$ $\overset{\mathrm{i.i.d.}}{\sim}$ $\mathcal{N}(0,1)$ and in Models 3 and 4, $Z_{j} \overset{\mathrm{i.i.d.}}{\sim} t_{3}/\sqrt{3}$. In Models 1, 3, 5, 6, 7, 9 and 10, the component populations differ only in their mean functions, while for the remaining models, the difference is in their covariance operators, specifically, in the structure of the eigenvalues. Differences in the decay rates of eigenvalues reflect varying levels of smoothness across the component populations. This makes clustering challenging, as the populations differ only in smoothness while their means remain unchanged. Models 3 and 4 involve $t$-processes to evaluate the robustness of the competing methods under heavy-tailed settings, while the remaining  models are based on Gaussian processes. The performance of Algorithm \ref{alg:grpm} and its competitors is reported in Table \ref{tab:comparison}.

\begin{table}[htbp]
\centering
\caption{Comparison of clustering methods across Models 1--10 (Equi-spaced time points)}\label{tab:comparison}
\resizebox{\textwidth}{!}{
\begin{tabular}{lccccccccccc}
\toprule
\textbf{Model} & \textbf{Sizes} & \textbf{TERP} & \textbf{PC} & \textbf{D$_\text{H}$} & \textbf{D$_\text{PC}$} & \textbf{F1} & \textbf{F2}&  \textbf{GPmix}& \textbf{sf}& \textbf{MC} \\
\midrule
1 & $(30,30)$  &       0.834&       0.703&       \textit{0.884}&  0.696&       0.501&        0.79&  0.511& \textbf{0.898}& 0.526\\
& $(30,70)$  &       0.81&       0.545&       \textit{0.848}&    0.671&   0.505&        0.508&  0.551& \textbf{0.912}&0.595\\
2 & $(30,30)$  &       \textbf{0.832} &       0.502&       \textit{0.506}&   \textit{0.506}&    0.492&       0.494&  0.497& 0.492& 0.492\\
& $(30,70)$  &       \textbf{0.828}&       0.534&       0.506&   0.512&    0.535&        0.56&  0.545& \textit{0.576}& \textit{0.576}\\
3 & $(30,30)$  &       \textit{0.752}&       0.504&       \textbf{0.8}&    0.573&   0.501&        0.506&  0.498& 0.65& 0.492\\
& $(30,70)$  &       \textit{0.74}&       0.535&       \textbf{0.775}&    0.577&   0.56&        0.593&  0.54& 0.736& 0.576\\
4 & $(30,30)$  &       \textbf{0.8}&       \textit{0.565}&       0.51&    0.533&      0.492 &        0.497&  0.499& 0.492& 0.492\\
& $(30,70)$  &       \textbf{0.855}&       \textit{0.612}&       0.513&    0.599&   0.593&        0.593&  0.578& 0.576& 0.576\\
5 & $(30,30)$  &       \textit{0.864}&       0.814&       0.731&    \textbf{0.959}&   0.765&        0.636&  0.539& 0.492& 0.506\\
& $(30,70)$  &       \textit{0.834}&       0.811&       0.644&    \textbf{0.949}&   0.677&        0.677&  0.594& 0.576& 0.619\\
6 & $(30,30)$  &       \textit{0.787}&       0.765&       0.734&    0.727&   0.79&        0.655&  0.556& \textbf{0.877}& 0.591\\
& $(30,70)$  &       0.718&       0.715&       \textit{0.729}&    0.707&   0.689&        0.621&  0.578& \textbf{0.858}& 0.657\\
7 & $(30,30,30)$  &       \textit{0.88}&       0.751&       \textbf{0.881}&    0.831&   0.724&        0.725&  0.489& 0.829& 0.567\\
& $(30,50,70)$  &       0.867&       0.708&       \textit{0.87}&    0.836&   0.679&        0.665&  0.482& \textbf{0.883}& 0.482\\
8 & $(30,30,30)$  &       \textbf{0.839}&       0.56&       0.564&    \textit{0.574}&   0.519&        0.55&  0.397& 0.326& 0.331\\
& $(30,50,70)$  &       \textbf{0.839}&       0.552&       0.554&    \textit{0.575}&   0.432&        0.512&  0.413& 0.365& 0.371\\
9 & $(30,30,30)$  &       \textbf{0.933}&       \textit{0.74}&       0.638&    0.696&   0.628&        0.686&  0.514& 0.326& 0.545\\
& $(30,50,70)$  &       \textbf{0.947}&       0.71&       0.62&   0.686&    \textit{0.728}&        0.669&  0.514& 0.368& 0.478\\
10 & $(30,30,30)$  &       \textit{0.841}&       0.548&       0.56&    0.561&   0.542&        0.539&  0.58& \textbf{0.891}& 0.708\\
& $(30,50,70)$  &       \textit{0.856}&       0.545&       0.548&    0.548&   0.521&        0.514&  0.622& \textbf{0.889}& 0.611\\
\bottomrule
\end{tabular}
}
\end{table}

It is evident from Table~\ref{tab:comparison} that Algorithm \ref{alg:grpm} attains either the best or the second-best performance across nearly all simulation settings. The only exceptions occur for Model 1 and the unbalanced sample size scenarios under Models 6 and 7. Even in these cases, its performance remains competitive and closely aligned with that of the leading method. In contrast, none of the competing procedures demonstrate such consistent performance across the diverse range of simulation models considered. In models characterized by scale differences, all competing methods fail. Similarly, under the heavy-tailed model with location differences, most methods fail, with the exceptions of \texttt{$\text{D}{\text{H}}$} and \texttt{sf}. Notably, the \texttt{$\text{D}_{\text{PC}}$} method, despite leveraging principal component directions, is unable to fully exploit this structural information. A particularly instructive comparison arises between Models 5 and 9, both of which involve location differences along principal component directions. While \texttt{$\text{D}_{\text{PC}}$} performs optimally in Model 5, its performance degrades significantly in Model 9, indicating a sub-optimal use of the information in the principal component directions. In contrast, Algorithm \ref{alg:grpm} effectively leverages the principal component structure in both settings, resulting in consistently strong performance. Although the \texttt{GPmix} method is based on random projections, it has comparatively weaker performance which indicates that it does not fully capitalize on the potential advantages of random projections unlike Algorithm \ref{alg:grpm}.

One of the reasons behind the strong performance of Algorithm \ref{alg:grpm} is its two-stage clustering framework, coupled with an ensemble strategy that aggregates results over multiple choices of Gaussian processes and number of projections. In particular, the second stage of clustering plays a crucial role in improving the average Rand index for Models 2, 4, and 8 (characterized by scale differences), as well as in several scenarios involving location differences.

\subsection{Analysis of real datasets: regular functional data} \label{real data analysis:reg}
We apply Algorithm \ref{alg:grpm} and the competing methods to eleven benchmark real datasets which are observed over regular grids. Ten of these are obtained from the UEA and UCR Time Series Classification Repository and the remaining one is the Meatspectrum dataset (spectrometric measurements of fat content in meat) available in the \texttt{faraway} package in CRAN. For clustering the Meatspectrum data, we have used the second order derivative curves as has been done in some of the existing literature (see \cite{ferraty2003curves, rossi2006support, li2008classification}). For the Berkeley Growth dataset and the Wheat dataset, clustering is performed using the first order derivative curves, which is consistent with existing literature (see \cite{slaets2012phase, delaigle2019clustering}).

\begin{table}[htbp]
\centering
\caption{Comparison of clustering methods across different datasets (Regular grid of time points)}\label{tab:comparison_real data}
\resizebox{\textwidth}{!}{
\begin{tabular}{lcccccccccc}
\toprule
\textbf{Dataset} & \textbf{K} & \textbf{TERP} & \textbf{PC} & \textbf{D\textsubscript{H}} & \textbf{D\textsubscript{PC}} & \textbf{F1} & \textbf{F2}&  \textbf{GPmix}& \textbf{sf}& \textbf{MC} \\
\midrule
Meatspec & 2  &       \textbf{0.889}&       \textit{0.862}&       0.759&  0.854&       0.701&        0.499&  0.708& -& 0.498\\
Coffee & 2  &       \textit{0.964}&       0.805&       0.514&   \textbf{1}&    \textbf{1}&        0.544&  0.494& -& 0.492\\
Flours & 3  &       \textbf{0.766}&       0.686&       0.715&    0.683&   \textit{0.747}&        0.658&  0.582& 0.688& 0.538\\
ECG & 2  &       \textbf{0.644}&       \textit{0.618}&  0.613     &   \textbf{0.644} &      0.554 &     0.574 &  0.599& 0.552&  0.52\\
Arrowhead &  3 &       \textit{0.625}&       0.596&       0.581&    0.588&   \textbf{0.685}&        0.623&  0.401& -& 0.53\\
Satellite &  2 &       \textbf{0.838}&       0.653&       0.635&    0.631&   0.534&        0.524&  \textit{0.707}& 0.62& 0.499\\
Berkeley & 2  &       \textbf{0.897}&       0.806&       \textit{0.878}&    \textit{0.878}&   0.613&        0.671&  0.511& 0.508& 0.496\\
Olive-oil &  4 &       0.777&       \textbf{0.845}&       0.718&    0.767&   \textit{0.783}&        -&  0.652& -& 0.288\\
Wheat & 2  &       \textit{0.728}&       \textit{0.728}&       0.715&    \textbf{0.729}&   0.519&        0.505& \textit{0.728}& -& 0.511\\
BME &  3 &       \textbf{0.833}&       0.611&       0.597&    0.609&   \textit{0.613}&        0.56&  0.564& 0.555& 0.515\\
Synthetic Control &  6 &       \textbf{0.925}&       \textit{0.87}&       0.849&    0.848&   0.643&        0.825&  0.768& 0.832& 0.761\\
\bottomrule
\end{tabular}
}
\end{table}

\begin{figure}
    \centering
    \includegraphics[scale=0.35]{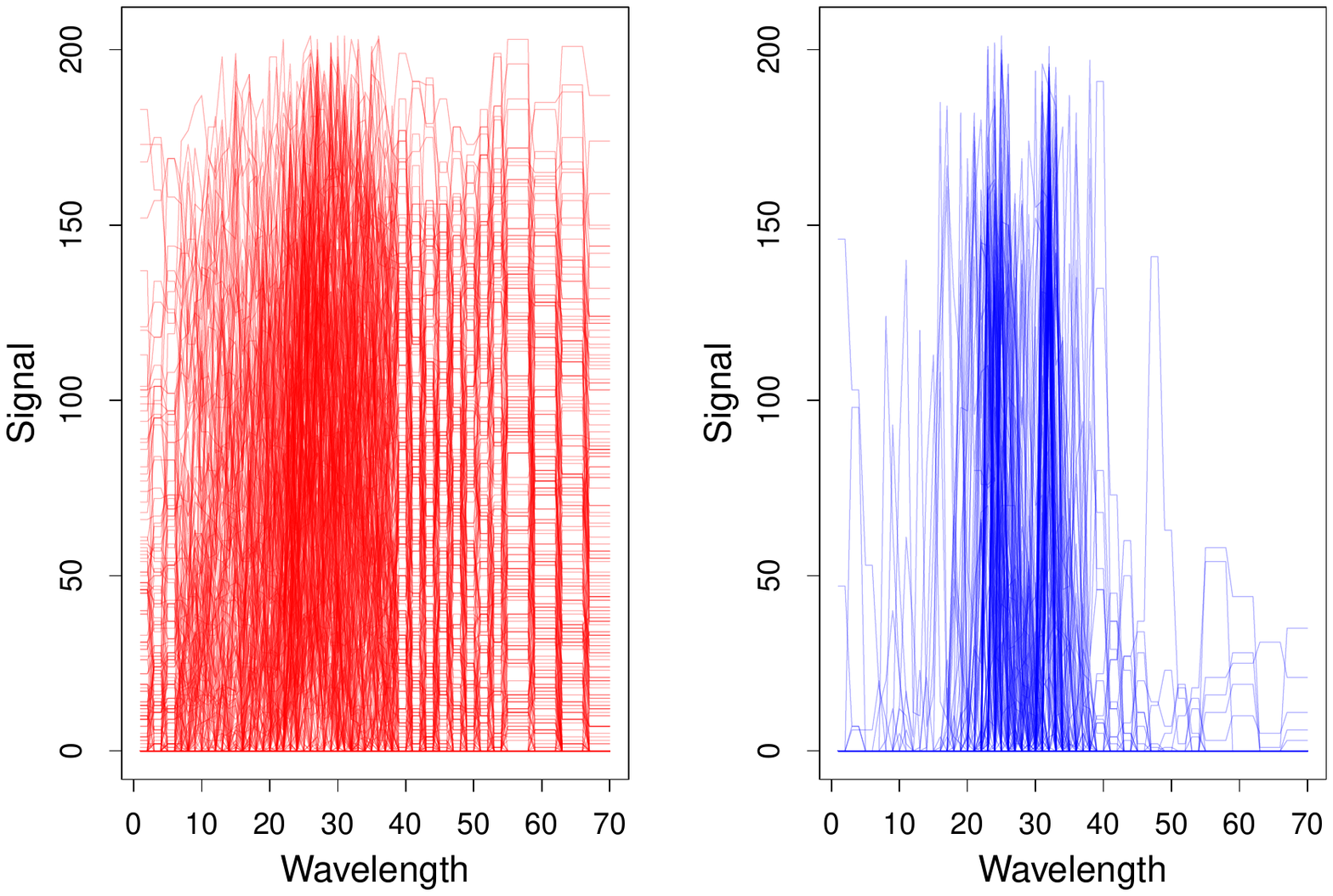}
    \includegraphics[scale=0.45]{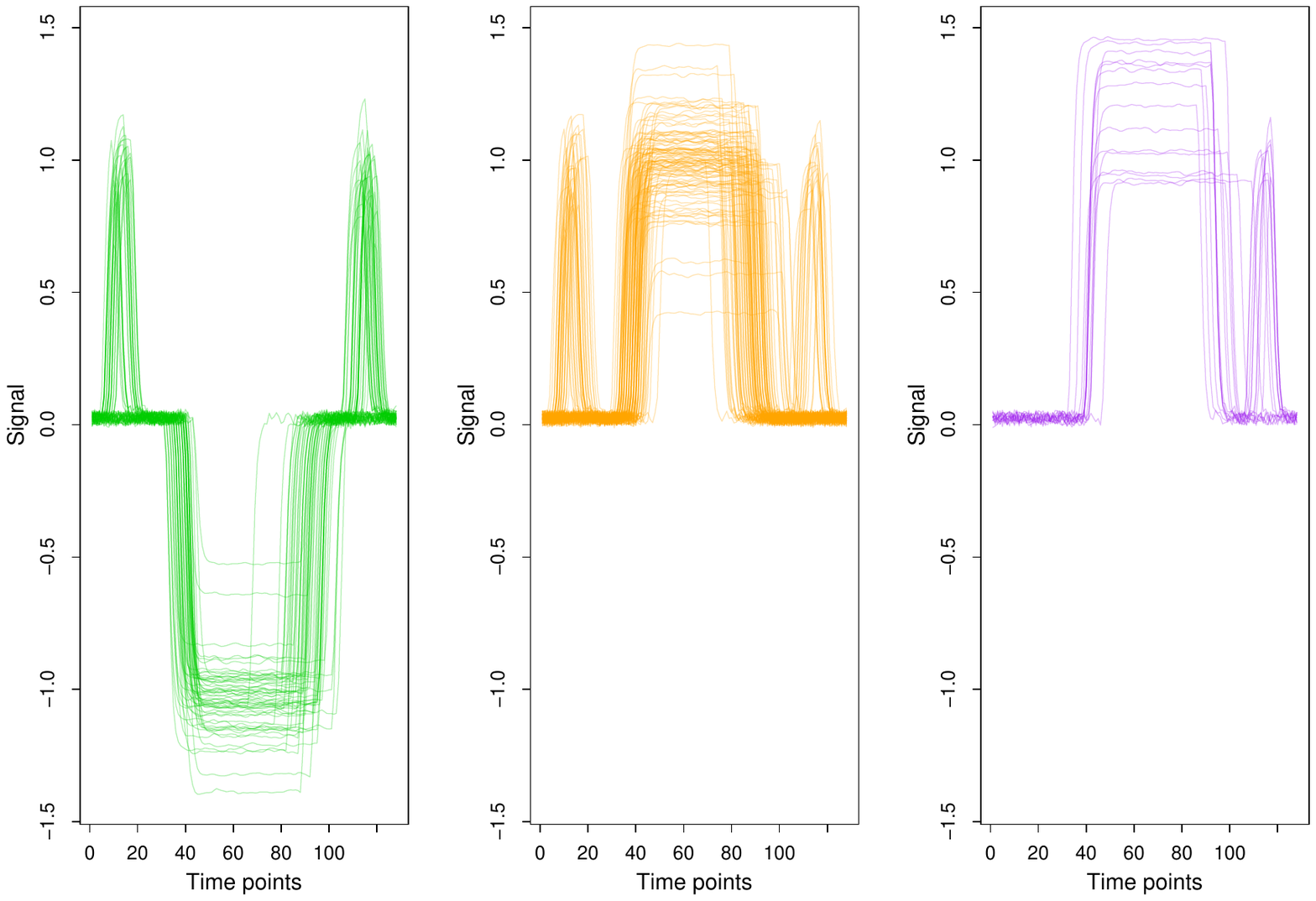}
    \caption{Plots of Satellite and BME datasets showing the different clusters obtained}
    \label{fig:regular-real-plots}
\end{figure}

Table~\ref{tab:comparison_real data} provides the Rand index for these eleven benchmark functional datasets which are observed on regular grids. The Rand index is computed by comparing the estimated class labels with the true class labels which are available for these datasets. Overall, Algorithm \ref{alg:grpm} has the best or the second best performance for all the datasets except the Olive Oil dataset, where its performance is very close to that of the second best method. As in the  simulation study, the second stage of clustering plays a crucial role in improving the average Rand index of Algorithm \ref{alg:grpm} in five out of the eleven real datasets considered. Moreover, it is observed that the best Rand indices are obtained for a variety of combinations of Gaussian processes and number of projections which clearly demonstrate the utility of the ensemble strategy in adapting to various types of functional data. For the purpose of demonstration, we have provided the plots of two datasets (Satellite and BME) where the clusters obtained by applying Algorithm \ref{alg:grpm} are shown using different colours (see Figure \ref{fig:regular-real-plots}).

\subsection{Analysis of simulated datasets: irregular and fragmented functional data} \label{simulation studies:irreg}
We next examine the performance of Algorithm \ref{alg:grpm_irreg} and some of its competitors under \emph{irregularly observed functional data}. In this setup, functional data are only recorded on discrete grids that differ across observations, and this could potentially affect clustering performance. For this study, we consider Models 1--10 used in Section \ref{simulation studies:reg} earlier and introduce irregularity in the following manner. The functional data from each of the models are initially recorded over a fine regular grid of $1000$ equi-distant time points over $[0,1]$. For each curve, we randomly select $100$ time points from this fine regular grid without replacement and the function values are retained at only these selected locations. As a result, each curve is observed on a curve-specific and randomly chosen set of $100$ time points, leading to irregular observation grids across subjects. Table~\ref{tab:simulation_irregular} reports the average Rand index for the competing clustering methods under this irregular grid setup.

\begin{table}[htbp]
\centering
\caption{Comparison of clustering methods across Models 1--10 (Irregular Data).}
\label{tab:simulation_irregular}
\resizebox{\textwidth}{!}{
\begin{tabular}{lcccccc}
\toprule
\textbf{Model} & \textbf{Sizes} & \textbf{TERPM} &  \textbf{PC}& \textbf{sf}& \textbf{MC} \\
\midrule
1 & $(30,30)$  &       \textit{0.811}&                0.784& \textbf{0.886}& 0.517\\
& $(30,70)$  &       \textit{0.794}&       0.759& \textbf{0.899}& 0.556\\
2 & $(30,30)$  &       \textbf{0.508}&               \textit{0.501}& 0.492& 0.492\\
& $(30,70)$  &       0.522&              0.505& \textbf{0.576}& \textit{0.575}\\
3 & $(30,30)$  &       \textbf{0.721}&               \textit{0.712}& 0.69& 0.492\\
& $(30,70)$  &       \textit{0.717}&                0.696& \textbf{0.773}& 0.576\\
4 & $(30,30)$  &       \textbf{0.596}&               \textit{0.504}& 0.492& 0.492\\
& $(30,70)$  &       \textbf{0.683}&               0.524& \textit{0.576}& \textit{0.576}\\
5 & $(30,30)$  &       \textbf{0.888}&               \textit{0.879}& 0.492& 0.494\\
& $(30,70)$  &       \textbf{0.87}&               \textit{0.849}& 0.576& 0.552\\
6 & $(30,30)$  &       0.778&               \textit{0.808}& \textbf{0.872}& 0.5\\
& $(30,70)$  &       0.722&               \textit{0.769}& \textbf{0.836}& 0.512\\
7 & $(30,30,30)$  &       \textbf{0.854}&               \textit{0.804}& 0.795& 0.473\\
& $(30,50,70)$  &       \textit{0.827}&               0.792& \textbf{0.865}& 0.474\\
8 & $(30,30,30)$  &       \textbf{0.65}&               \textit{0.547}& 0.326& 0.33\\
& $(30,50,70)$  &       \textbf{0.588}&               \textit{0.545}& 0.365& 0.373\\
9 & $(30,30,30)$  &       \textbf{0.899}&               \textit{0.87}& 0.326& 0.457\\
& $(30,50,70)$  &       \textbf{0.89}&               \textit{0.844}& 0.365& 0.488\\
10 & $(30,30,30)$  &       \textit{0.823}&               0.808& \textbf{0.867}& 0.561\\
& $(30,50,70)$  &       \textit{0.838}&               0.818& \textbf{0.882}& 0.542\\
\bottomrule
\end{tabular}
}
\end{table}

It is observed from Table \ref{tab:simulation_irregular} that there is a drop in the clustering accuracy from the setup of regular grids which is expected. It is observed that Algorithm \ref{alg:grpm_irreg} outperforms its competitors and in fact has the best or the second best performance in almost all scenarios except Models 2 and 6. For Model 6, our method has competitive performance. However, for Model 2, our method and all the other competing methods perform poorly.

We also compare the performance of the clustering algorithms in the setting of \emph{fragmented functional data}. In this setup, observations are unavailable over certain sub-intervals of the domain, which could potentially affect the clustering performance like the irregular observation grid setup considered earlier. For the comparative study, we use Models 1--10 considered earlier and incorporate a mechanism which creates fragmented functional data for each of these models. Specifically, let the complete observation domain $\mathcal{T} = [0,1]$ be partitioned into $R=10$ disjoint sub-intervals of equal length, $\mathcal{S}_1,\mathcal{S}_2,\dots,\mathcal{S}_{10}$. For each subject $i$, an index $r_i$ is sampled uniformly from $\{1,2,\dots,10\}$ and all observation points within the segment $\mathcal{S}_{r_i}$ are removed. The observed domain for the $i$th curve is therefore given by $\mathcal{T}_i = \mathcal{T} \setminus \mathcal{S}_{r_i}$. This mechanism induces a missingness rate of $10\%$ per curve, introducing a single contiguous gap at a random location along the domain. Table~\ref{tab:simulation_fragmented} reports the average Rand index for the competing clustering methods under this fragmentation scheme.

\begin{table}[h!]
\centering
\caption{Comparison of clustering methods across Models 1--10 (Fragmented Data).}
\label{tab:simulation_fragmented}
\resizebox{\textwidth}{!}{
\begin{tabular}{lcccccc}
\toprule
\textbf{Model} & \textbf{Sizes} & \textbf{TERPM} & \textbf{PC}& \textbf{sf}& \textbf{MC}  \\
\midrule
1 & $(30,30)$  &       \textit{0.791}&    0.621 & \textbf{0.869} & 0.517         \\
& $(30,70)$  &       \textit{0.766}&       0.62& \textbf{0.888}& 0.591\\
2 & $(30,30)$  &       \textbf{0.664}&        \textit{0.514}& 0.492& 0.492       \\
& $(30,70)$  &       \textbf{0.704}&       0.505& \textit{0.576}& 0.575\\
3 & $(30,30)$  &       \textbf{0.699}&               0.546& \textit{0.658}& 0.492\\
& $(30,70)$  &       \textbf{0.758}&       0.549& \textit{0.733}& 0.576\\
4 & $(30,30)$  &       \textbf{0.666}&               \textit{0.512}& 0.492& 0.492\\
& $(30,70)$  &       \textbf{0.751}&       \textit{0.581}& 0.576& 0.576\\
5 & $(30,30)$  &       \textbf{0.78}&               \textit{0.67}& 0.492& 0.516\\
& $(30,70)$  &       \textbf{0.744}&       \textit{0.679}& 0.576& 0.612\\
6 & $(30,30)$  &       0.777&               \textit{0.795}& \textbf{0.916}& 0.603\\
& $(30,70)$  &       0.73&       \textit{0.75}& \textbf{0.885}& 0.671\\
7 & $(30,30,30)$  &       \textbf{0.824}&               0.753& \textit{0.805}& 0.576\\
& $(30,50,70)$  &       \textit{0.808}&       0.749& \textbf{0.849}& 0.493\\
8 & $(30,30,30)$  &       \textbf{0.722}&               \textit{0.507}& 0.326& 0.329\\
& $(30,50,70)$  &       \textbf{0.655}&       \textit{0.534}& 0.365& 0.376\\
9 & $(30,30,30)$  &       \textbf{0.827}&               \textit{0.774}& 0.326& 0.539\\
& $(30,50,70)$  &       \textbf{0.807}&       \textit{0.775}& 0.373& 0.486\\
10 & $(30,30,30)$  &       \textit{0.826}&               0.823& \textbf{0.906}& 0.721\\
& $(30,50,70)$  &      \textit{0.845} &       0.828& \textbf{0.902}& 0.611\\
\bottomrule
\end{tabular}
}
\end{table}

It is observed from Table \ref{tab:simulation_fragmented} that there is a drop in the clustering accuracy from the setup of regular observation grids which is expected. The proposed \ref{alg:grpm_irreg} algorithm is either the best or the second best performer in all scenarios except Model 6, where it is still a strong performer.

\subsection{Analysis of real datasets: irregular and fragmented functional data}
We now investigate the performance of the proposed Algorithm \ref{alg:grpm_irreg} and its competitors for irregularly observed and fragmented real datasets. We construct irregular versions of the real datasets considered in Section \ref{real data analysis:reg} in the following manner. For a given curve which is observed over $T$ time points, say $\{t_1,\ldots,t_T\}$, we select $\lfloor 0.1T \rfloor$ time points uniformly at random without replacement and remove the corresponding function values. The resulting curves are thus observed on curve-specific subsets of the original time grid, producing irregular observation patterns across subjects. We also consider the \emph{Medfly} dataset, which is naturally observed on irregular and subject-specific time grids. This dataset was collected by \cite{carey1998medfly} in a large-scale bio-demographic study at the University of California, Davis. The data consist of daily egg counts for approximately 1,000 individually monitored female flies, recorded from emergence until death. Each fly yields a dense longitudinal trajectory observed over a subject-specific lifespan, resulting in dense but irregular functional observations across individuals. This dataset is widely used in functional data analysis, particularly for clustering and classification under irregular designs. The original dataset is publicly available at \url{http://anson.ucdavis.edu/~mueller/data/medfly1000.html}.

\begin{table}[h!]
\centering
\caption{Comparison of clustering methods across real data (irregular grid of time points).}
\label{tab:comparison_irregular_real_data}
\resizebox{\textwidth}{!}{
\begin{tabular}{lcccccc}
\toprule
\textbf{Dataset} & \textbf{K} & \textbf{TERPM} & \textbf{PC}& \textbf{sf}& \textbf{MC}  \\
\midrule
Meatspec & 2  &       \textbf{0.83}&    \textit{0.752}& -           &0.501\\
Coffee& 2& \textbf{0.584}& \textit{0.494}& -& 0.492\\
Flours & 3  &       \textbf{0.768}&        0.686& \textit{0.688}&-       \\
ECG & 2  &       \textbf{0.655}&               \textit{0.608}& 0.552& 0.499\\
Arrowhead & 3  &       \textit{0.584}&               \textbf{0.585}& -& 0.491\\
Satellite & 2  &       \textbf{0.688}&               \textit{0.649}& 0.624& 0.5\\
Berkeley & 2  &       \textbf{0.806}&               \textit{0.789}& 0.624& 0.499\\
Olive-oil & 4  &       \textbf{0.766}&               \textit{0.614}&-& 0.288\\
Wheat & 2  &       \textbf{0.689}&               \textbf{0.689}& -& \textit{0.511}\\
BME & 3  &       \textbf{0.615}&               \textit{0.555}& \textit{0.555}& 0.551\\
Synthetic Control & 6  &       \textit{0.838}&               \textbf{0.841}& 0.831& 0.542\\
Medfly & 2& \textbf{0.766}& \textit{0.637}& 0.5& 0.5\\
\bottomrule
\end{tabular}
}
\end{table}

Table~\ref{tab:comparison_irregular_real_data} presents the Rand indices for Algorithm \ref{alg:grpm_irreg} and its competitors on the real functional datasets considered above that are observed on irregular grids. It is seen that the proposed method is the best performer for ten out of the twelve datasets considered, while it is the second best performer in the remaining two datasets (Arrowhead and Synthetic Control datasets). For some competing methods, there were instances when they did not converge or did not return any result. These cases are marked with a hyphen in Table \ref{tab:comparison_irregular_real_data}. For the purpose of demonstration, we have provided the plots of two datasets (Berkeley and Medfly) where the clusters obtained by applying Algorithm \ref{alg:grpm_irreg} are shown using different colours (see Figure \ref{fig:irregular-real-plots}).

\begin{figure}
    \centering
    \includegraphics[scale=0.35]{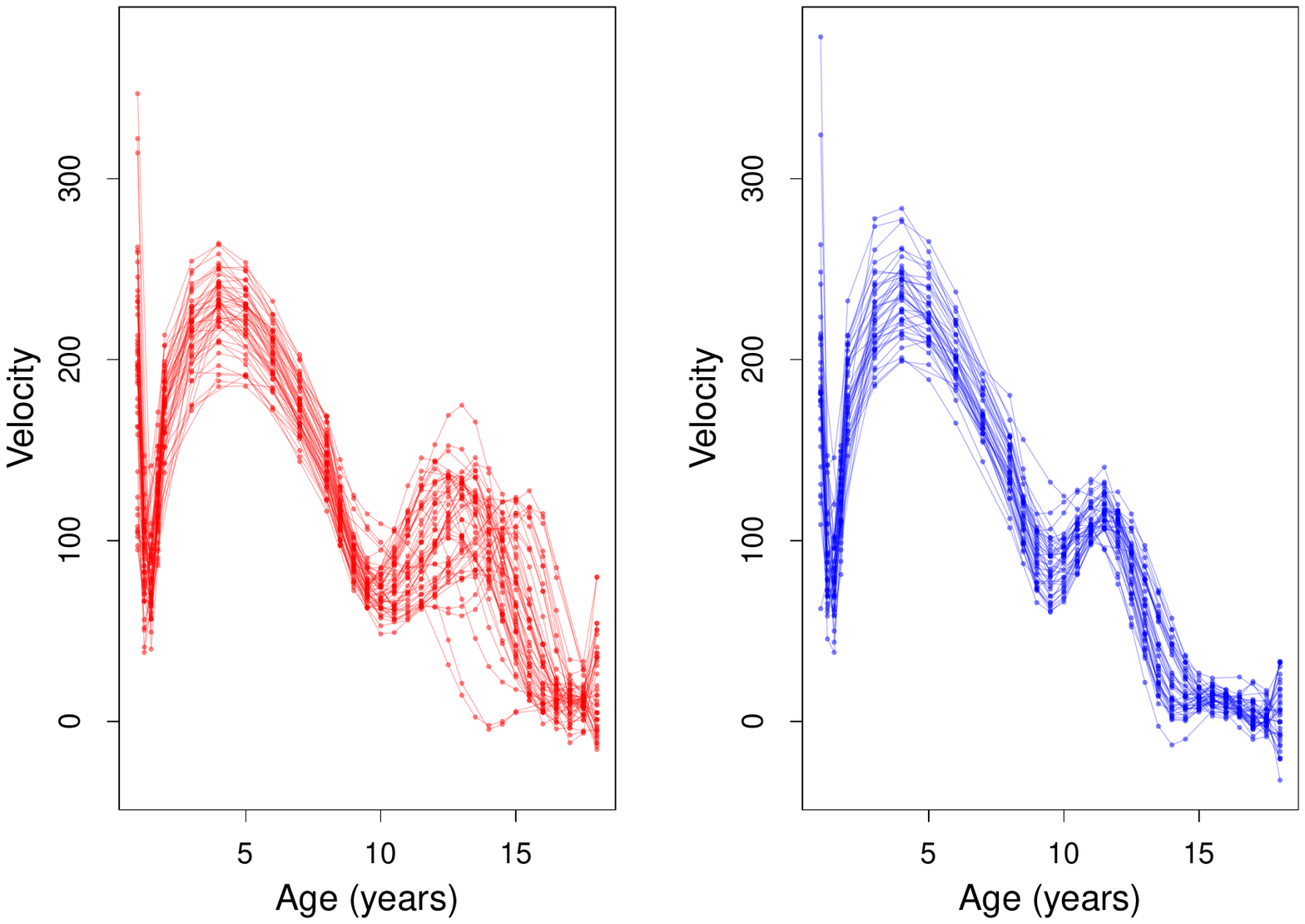}
    \includegraphics[scale=0.45]{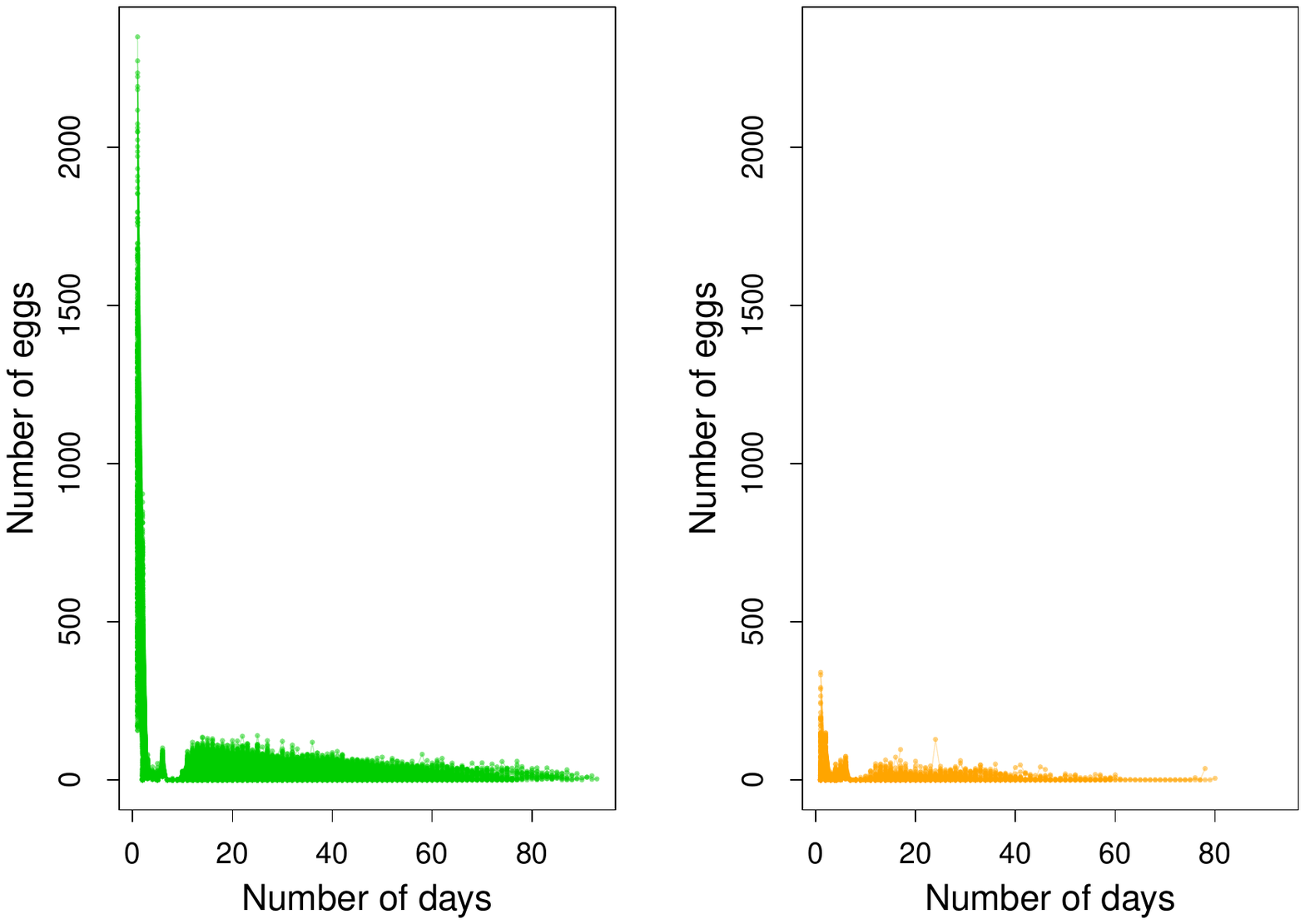}
    \caption{Plots of Berkeley and Medfly datasets showing the different clusters obtained}
    \label{fig:irregular-real-plots}
\end{figure}

We have also investigated the performance of the different clustering methods for fragmented real datasets. For the real datasets considered in Section \ref{real data analysis:reg}, we have introduced fragmentation of the observation grid in the same way as in Section \ref{simulation studies:irreg} earlier. Table \ref{tab:comparison_fragmented_real_data} provides the Rand indices obtained in this setup. It is observed that Algorithm \ref{alg:grpm_irreg} has the best performance for all the datasets except the Coffee and the Wheat data. For these two datasets, all of the methods have poor performance. For the purpose of demonstration, we have provided the plots of two datasets (Meatspectrum and Flours) where the clusters obtained by applying Algorithm \ref{alg:grpm_irreg} are shown using different colours (see Figure \ref{fig:fragmented-real-plots}).

\begin{figure}
    \centering
    \includegraphics[scale=0.35]{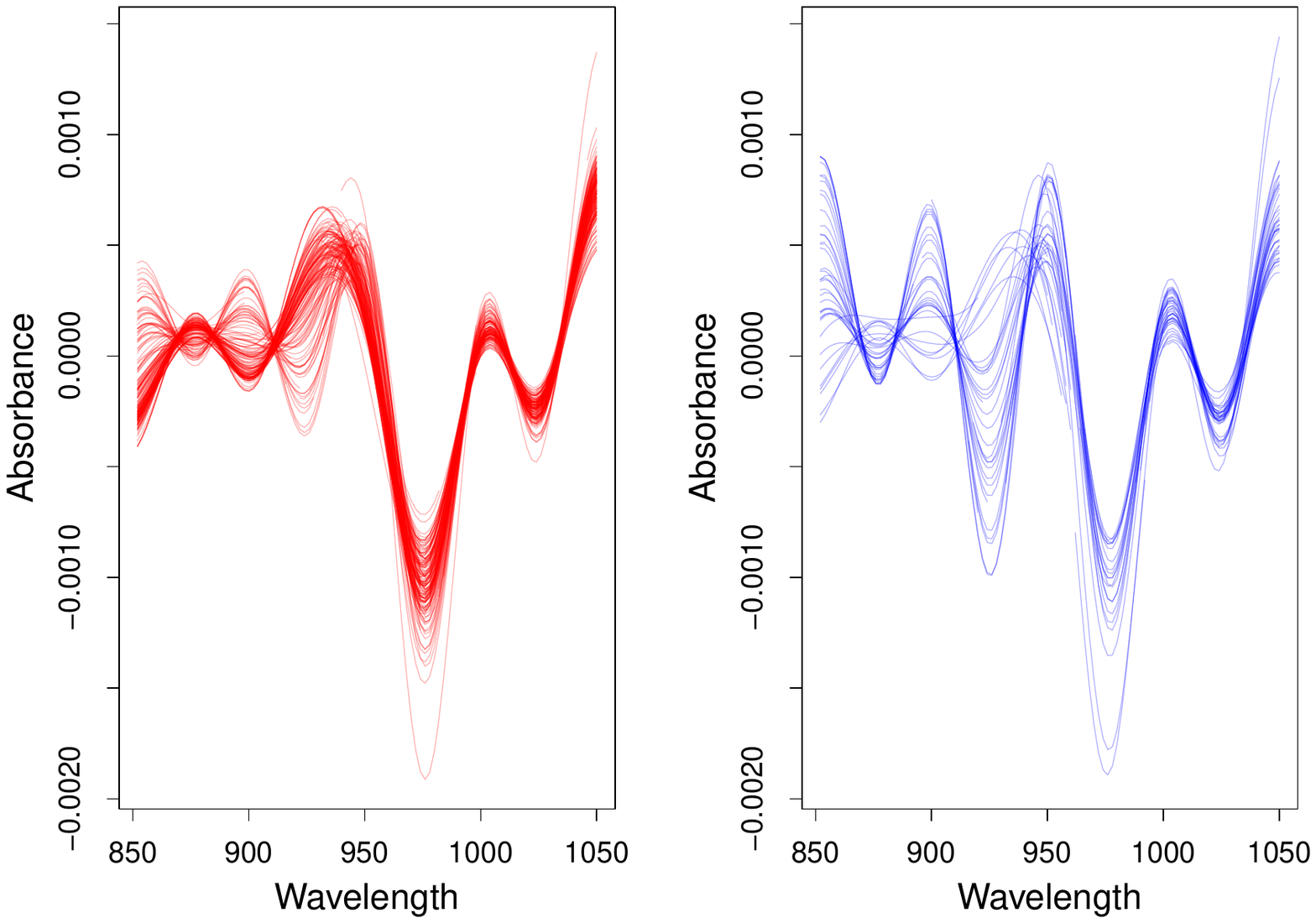}
    \includegraphics[scale=0.45]{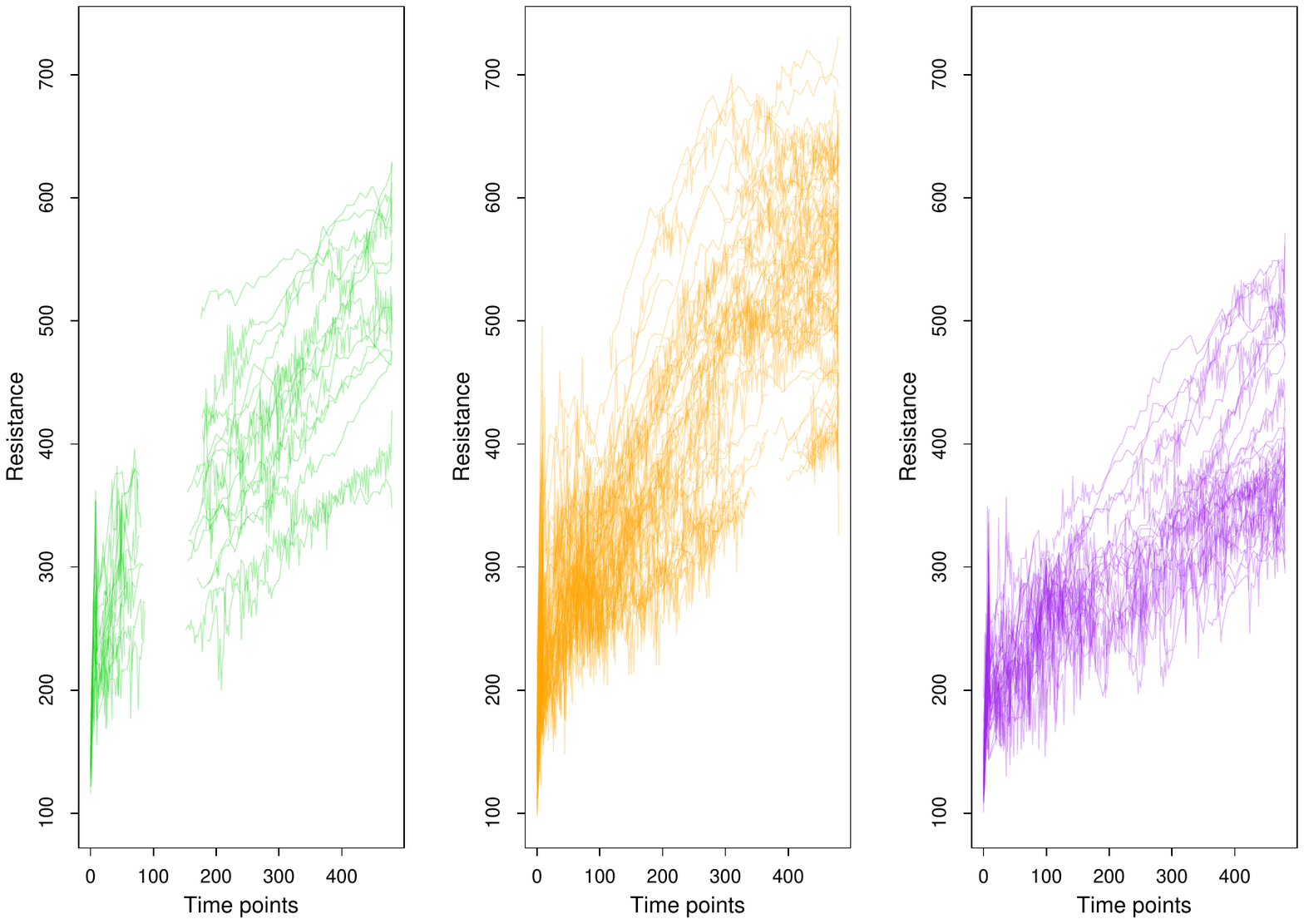}
    \caption{Plots of Meatspectrum and Flours datasets (manually fragmented) showing the different clusters obtained}
    \label{fig:fragmented-real-plots}
\end{figure}

\begin{table}[h!]
\centering
\caption{Comparison of clustering methods across fragmented real data (manually introduced fragmentation).}
\label{tab:comparison_fragmented_real_data}
\resizebox{\textwidth}{!}{
\begin{tabular}{lcccccc}
\toprule
\textbf{Dataset} & \textbf{K} & \textbf{TERPM} & \textbf{PC}& \textbf{sf}& \textbf{MC}  \\
\midrule
Meatspec & 2  &       \textbf{0.903}&    \textit{0.759} &- & 0.538         \\
Coffee& 2& \textbf{0.494}& \textbf{0.494}&- & \textit{0.492}\\
Flours & 3  &       \textbf{0.769}&        \textit{0.685}& 0.677&    -   \\
ECG& 2& \textbf{0.595}& 0.512& \textit{0.552}& 0.516\\
Arrowhead & 3  &       \textbf{0.634}&               \textit{0.585}& -& 0.512\\
Satellite & 2  &       \textbf{0.675}&               \textit{0.653}& 0.633& 0.499\\
Berkeley& 2& \textbf{0.727}& \textit{0.698}& 0.508& 0.516\\
Olive-oil & 4  &       \textbf{0.838}&               \textit{0.588}& -& 0.288\\
Wheat & 2&  0.496& \textbf{0.547}& -& \textit{0.511}\\
BME & 3& \textbf{0.633}& \textit{0.555}& \textit{0.555}& 0.496\\
Synthetic Control & 6  &       \textit{0.842}&               \textbf{0.862}& 0.832& 0.745\\
\bottomrule
\end{tabular}
}
\end{table}

\subsection{Computation time}
We now report the computation time of the proposed Algorithm \ref{alg:grpm} and Algorithm \ref{alg:grpm_irreg}. The goal of this study is to investigate the effect of the number of projection directions, number of populations and the total sample size on the computational costs of the algorithms. All experiments were performed on a computer equipped with a 12th Gen Intel(R) Core(TM) i7-12700 (2.10 GHz) and 16 GB RAM running a 64-bit operating system. The proposed algorithm was implemented in R. Functional observations were generated on a grid of 100 time points over $[0,1]$ (equi-spaced for regular functional data and randomly chosen for irregular functional data). Each reported value in Table \ref{tab:timing} represents the computation time (in seconds) for a single iteration of the clustering algorithm.

\begin{table}[htbp]
\centering
\caption{Computing times (in seconds) of the \ref{alg:grpm} algorithm and those for the \ref{alg:grpm_irreg} algorithm (provided in brackets) for different values of total sample size ($N$), number of clusters ($K$), and random projections ($M$).}
\label{tab:timing}
\resizebox{\textwidth}{!}{
\begin{tabular}{lccccccc}
\toprule
\multirow{2}{*}{$N$} &
\multirow{2}{*}{$K$} &
\multirow{2}{*}{$n$} &
\multicolumn{5}{c}{Number of projections ($M$)} \\
\cmidrule(lr){4-8}
& & & 10 & 50 & 100 & 500 & 1000 \\
\midrule
\multirow{6}{*}{$200$}
  & $2$  & $100$ &  \text{1.39} &  \text{1.84} &   \text{2.50} &   \text{8.09} &    \text{14.95} \\
  & & & \text{(68.33)} & \text{(69.17)} & \text{(70.89)} & \text{(88.86)} & \text{(110.23)} \\
  & $5$  & $40$  &  \text{1.28} &  \text{1.86} &   \text{2.56} &   \text{8.04} &    \text{15.10} \\
    & & & \text{(67.17)} & \text{(69.24)} & \text{(71.17)} & \text{(88.69)} & \text{(110.62)} \\
  & $10$ & $20$  &  \text{1.17} &  \text{1.75} &   \text{2.33} &   \text{7.58} &    \text{14.12} \\
  & & & \text{(66.11)} & \text{(68.31)} & \text{(70.05)} & \text{(88.40)} & \text{(109.85)} \\
\midrule
\multirow{6}{*}{$1000$}
  & $2$  & $500$ &  \text{31.64} &  \text{43.47} &  \text{58.32} & \text{175.61} &  \text{322.39} \\
    & & & \text{(277.75)} & \text{(294.88)} & \text{(316.81)} & \text{(500.86)} & \text{(727.23)} \\
  & $5$  & $200$ &  \text{29.73} &  \text{41.20} & \text{55.78} & \text{171.72} &   \text{315.86} \\
    & & & \text{(276.69)} & \text{(292.14)} & \text{(313.92)} & \text{(494.43)} & \text{(719.21)} \\
  & $10$ & $100$ &  \text{28.26} &  \text{39.21} &  \text{52.74} & \text{162.19} &   \text{299.69} \\
    & & & \text{(274.12)} & \text{(290.55)} & \text{(314.23)} & \text{(487.19)} & \text{(699.11)} \\
\bottomrule
\end{tabular}
}
\end{table}

\section{Conclusion}
\label{conclusion}
This paper proposes a simple and flexible projection-based framework for clustering functional data, where a large collection of random projections is used to effectively reduce the infinite-dimensional clustering problem to a high-dimensional one. A notable advantage of the proposed methodology is its favorable computational complexity, ranging from linear to at most quadratic in the sample size, which makes the algorithms computationally efficient and easy to implement in practice. Another important strength of the proposed framework is its ability to accommodate a wide range of observational settings for functional data, including regular, irregular, and fragmented observations. The proposed algorithms are implemented for a prespecified number of clusters $K$. In practice, one may incorporate an additional optimization layer over different choices of $K$ by running the algorithms across a range of candidate values of $K$ and selecting the value that minimizes the corresponding optimal cost functions in \eqref{min-cost} for regular functional data and \eqref{min-cost-irreg} for irregular and fragmented functional data. Another layer of ensembling, in addition to the choice of Gaussian projection directions and their number, which may be incorporated in the proposed clustering methods would be with respect to the choice of kernels in the definition of the MADD dissimilarity measure (the paper presently uses the Laplace kernel). Although we do not investigate this perspective in the paper, one could use a collection of characteristic kernels on the real line for this purpose. The empirical studies indicate that the performance of the proposed methods deteriorates to some extent when the observation grids become highly sparse or when the functional trajectories contain substantial missingness over the domain. This phenomenon is likely due to reduced accuracy in approximating the Riemann integrals involved in the projection step. Addressing this limitation for sparse and highly fragmented functional data remains an important direction for future research.

\bibliographystyle{plain}
\bibliography{cas-refs}

\end{document}